\documentclass[
  aps,
  prl,
  reprint,
  superscriptaddress,
  longbibliography,
  nofootinbib
]{revtex4-2}

\usepackage{
  amsmath,
  amssymb,
  amsthm,
  mathtools,
  bm,
  booktabs,
  array
}
\usepackage{tikz}
\usepackage{hyperref}
 \hypersetup{colorlinks,linkcolor={blue},citecolor={red},urlcolor={blue}}
\usepackage{tikz,float,multirow,booktabs,enumerate}
\newcommand{\AME}{\operatorname{AME}}
\newcommand{\Tr}{\operatorname{Tr}}
\newcommand{\Z}{\mathbb Z}
\newcommand{\F}{\mathbb F}
\newcommand{\C}{\mathbb C}
\newcommand{\one}{\mathbf 1}
\newcommand{\im}{\operatorname{im}}
\newcommand{\supp}{\operatorname{supp}}
\newcommand{\wt}{\operatorname{wt}}

\newcommand{\zero}{\mathbf 0}

\newtheorem{theorem}{Theorem}

\begin{document}
\hbadness=10000

\title{Complete Existence Classification of Seven-Partite Absolutely Maximally Entangled States}
\author{Fei Shi}
   \email[]{shif26@mail.sysu.edu.cn}
	\affiliation{School of Computer Science and Engineering, Sun Yat-sen University, Guangzhou 510006, China}	

\author{Xiande Zhang}
   \email[]{drzhangx@ustc.edu.cn}
\affiliation{School of Mathematical Sciences,
		University of Science and Technology of China, Hefei, 230026, China; and with Hefei National Laboratory, University of Science and Technology of China, Hefei 230088, China}
        
\author{Qi Zhao}
\email[]{zhaoqi@cs.hku.hk}
 \affiliation{QICI Quantum Information and Computation Initiative, School of Computing and Data Science,
The University of Hong Kong, Pokfulam Road, Hong Kong SAR, China}

\author{Lvzhou Li}
\email[]{lilvzh@mail.sysu.edu.cn} 
	\affiliation{School of Computer Science and Engineering, Sun Yat-sen University, Guangzhou 510006, China}


\begin{abstract}
We prove that an absolutely maximally entangled state of seven qudits
exists if and only if the local dimension satisfies $d\geq 3$. Prior to
this work, to the best of our knowledge, $\operatorname{AME}(7,d)$
states were known to exist only when $d$ is a prime power other than
$2$, or when $d$ can be expressed as a product of dimensions for which
existence was already known. Since it has been proved that no
$\operatorname{AME}(7,2)$ state exists, it remains to establish
existence for all $d\geq 3$. We construct cyclic quadratic-phase states
for every odd local dimension and develop a coupled
binary--odd-dimensional construction for every dimension congruent to
$2$ modulo $4$. Together with the known power-of-two cases and the
product property of AME states, these constructions cover every local
dimension $d\geq 3$.
\end{abstract}
\maketitle

\textit{\textbf{Introduction.}}\textbf{---}
A pure state is called absolutely maximally entangled (AME) if it is maximally entangled across every bipartition, equivalently, if every reduced state involving at most half of the parties is maximally mixed~\cite{Scott2004,Helwig2012}. AME states correspond to pure quantum MDS codes~\cite{Scott2004,HuberGrassl2020}: specifically, an $n$-qudit AME state, denoted by $\AME(n,d)$, is equivalent to a pure quantum MDS code with parameters $[[n,0,\lfloor n/2\rfloor+1]]_d$. AME states have found applications in quantum secret sharing and parallel teleportation~\cite{Helwig2012}, holographic quantum error-correcting codes~\cite{Pastawski2015,Mazurek2020}, and quantum repeaters~\cite{AlsinaRazavi2021}. The existence problem of AME states has therefore attracted considerable attention~\cite{HiguchiSudbery2000,Feng2017,Huber2017,PRXQuantum.3.010101,Rather2022,AMEtable,RajchelReview2026}.

For example, Rather \emph{et al.} constructed an $\AME(4,6)$ state from entangled quantum Latin squares~\cite{Rather2022}, while Ball and Simoens proved that product-orthogonal quantum Latin squares of order six do not exist, demonstrating that entanglement is essential for this construction~\cite{BallSimoens2026}. For seven parties, $\AME(7,2)$ is known not to exist~\cite{Huber2017}; however, a complete existence classification of $\AME(7,d)$ states was still lacking.  To the best of our knowledge, prior to this work, the existence of AME states had been established only for local dimensions $d$ that are prime powers other than $2$, or that can be expressed as products of dimensions for which AME states are already known to exist. The remaining dimensions were unresolved; in particular, the existence of $\AME(7,6)$ and $\AME(7,10)$ was listed as open in Refs.~\cite{AMEtable,RajchelReview2026}.

In this Letter, we establish the complete classification
\begin{equation}
\boxed{\AME(7,d)\ \text{exists if and only if}\ d\geq3.}
\label{eq:main}
\end{equation}
The key ingredients are explicit uniform constructions that work in every odd dimension and every dimension congruent to $2$ modulo $4$.


\textit{\textbf{Notation.}}\textbf{---}
All subsets of $\Z_7$ are written in increasing order unless a translated order is explicitly specified.  Vectors indexed by a subset $A$ are row vectors; $v_A$ denotes restriction to $A$, and $\zero_A$ denotes the zero row vector on $A$.  For subsets $A,B\subseteq\Z_7$, $\Gamma_{A,B}$ is the submatrix with rows indexed by $A$ and columns indexed by $B$.  We write $\supp(v)$ and $\wt(v)$ for the support and Hamming weight of a binary vector, respectively.  Compact strings such as $013$ denote the ordered subset $(0,1,3)$.  Scalar zero is written $0$, whereas zero vectors and zero operators are written $\zero$ with a subscript when useful.

\textit{\textbf{Odd local dimensions.}}\textbf{---}
Label the parties by $\Z_7$ and let $\Gamma$ be the adjacency matrix of the complement of the seven-cycle,
\begin{equation}
 \Gamma_{ij}=\begin{cases}
 1,&i-j\not\equiv0,\pm1\pmod7,\\
 0,&i-j\equiv0,\pm1\pmod7.
 \end{cases}
 \label{eq:Gamma}
\end{equation}
In the ordered basis $\Z_7=(0,1,\ldots,6)$, this matrix is
\[
\Gamma=\begin{pmatrix}
0&0&1&1&1&1&0\\
0&0&0&1&1&1&1\\
1&0&0&0&1&1&1\\
1&1&0&0&0&1&1\\
1&1&1&0&0&0&1\\
1&1&1&1&0&0&0\\
0&1&1&1&1&0&0
\end{pmatrix}.
\]
Let $\widetilde\Gamma=(\widetilde\Gamma_{i,j})$ with $\widetilde\Gamma_{i,j}=\Gamma_{i,j}$ for $i<j$ and $0$ otherwise.
For the row vector $x=(x_0,\ldots,x_6)$, define
\begin{equation}
 Q(x)=x\widetilde\Gamma x^{\mathsf T}
 =\sum_{0\leq i<j\leq6}\Gamma_{ij}x_ix_j\pmod q.
 \label{eq:Q}
\end{equation}
and define
\begin{equation}
 |G_q\rangle=q^{-7/2}\sum_{x\in\Z_q^7}\omega_q^{Q(x)}|x\rangle,
 \qquad \omega_q=e^{2\pi i/q}.
 \label{eq:oddstate}
\end{equation}
According to Theorem~3 of Ref.~\cite{Feng2017}, $\Tr_{S^c}|G_q\rangle\langle G_q|=I_{q^3}/q^3$ for every three-subset $S$ if, for every three-subset $S$, some $3\times 3$ submatrix $\Gamma_{S,R}$ with $R\subseteq S^c$ is invertible over $\Z_q$. Then we need to find $35$ such submatrices. Since \(\Gamma_{S+t,R+t} =\Gamma_{S,R}\) holds for all \(t\in\mathbb{Z}_7\), where $S+t=(s+t)_{s\in S}$, it suffices to consider five representative sets $S$: $012$, $013$, $014$, $015$, and $024$. Assigning the corresponding sets $R$ to be $346$, $245$, $236$, $234$ and $135$ respectively, the matrices \(\Gamma_{S,R}\) are given by
\begin{equation}
\left(\begin{smallmatrix}1&1&0\\1&1&1\\0&1&1\end{smallmatrix}\right),\
\left(\begin{smallmatrix}1&1&1\\0&1&1\\0&0&1\end{smallmatrix}\right),\
\left(\begin{smallmatrix}1&1&0\\0&1&1\\1&0&1\end{smallmatrix}\right),\
\left(\begin{smallmatrix}1&1&1\\0&1&1\\1&1&0\end{smallmatrix}\right),\
\left(\begin{smallmatrix}0&1&1\\0&0&1\\1&0&0\end{smallmatrix}\right).
\end{equation}
Their determinants are $-1$, $1$, $2$, $-1$ and $1$ in the listed order. Consequently, each \(\Gamma_{S,R}\) is invertible over \(\mathbb{Z}_q\), which implies that 
\begin{equation}
 |G_q\rangle\ \text{is an }\AME(7,q)\text{ state}\qquad(q\geq3\text{ odd}).
 \label{eq:oddresult}
\end{equation}

\textit{\textbf{Twice-odd local dimensions.}}\textbf{---}
Let $C\subseteq\Z_2^7$ be the $[7,3,4]$ simplex code generated by
\begin{equation}
H=\begin{pmatrix}
1&0&1&1&1&0&0\\
1&1&1&0&0&1&0\\
0&1&1&1&0&0&1
\end{pmatrix},
\qquad C=\{mH:m\in\Z_2^3\}.
\label{eq:Gsim}
\end{equation}
The logical basis of the Steane code~\cite{Steane1996} is
\begin{equation}
 |0_L\rangle=8^{-1/2}\sum_{u\in C}|u\rangle,
 \qquad
 |1_L\rangle=8^{-1/2}\sum_{u\in C}|u+\one\rangle.
 \label{eq:logical}
\end{equation}
The eight simplex words and their complementary coset representatives are
listed in Table~\ref{tab:binarywords}.  
\begin{table}[t]
\caption{The simplex code $C$ and the coset $C+\one$. Coordinates are
indexed by $0,\ldots,6$.}
\label{tab:binarywords}
\centering\small
\setlength{\tabcolsep}{5.0pt}
\begin{ruledtabular}
\begin{tabular}{c c|c c}
$u$ & $u+\one$ & $u$ & $u+\one$\\
\hline
0000000&1111111&0010111&1101000\\
0111001&1000110&1001011&0110100\\
1110010&0001101&1100101&0011010\\
1011100&0100011&0101110&1010001
\end{tabular}
\end{ruledtabular}
\end{table}
The union $D=C\cup(C+\one)$ is the $[7,4,3]$ Hamming code, and
$D^\perp=C$ with distance $4$. 
Fix a three-subset $S$ and put
$B=S^c$.  For the same cut $S|B$, define the logical reduction blocks
\begin{equation}
 R_{ab}^S:=\Tr_B|a_L\rangle\langle b_L|,
 \qquad a,b\in\Z_2.
 \label{eq:logicalblocks}
\end{equation}
Since the distance of $D^\perp$ is $4$, the projection of $D$ onto any three coordinates
is uniform: each element of $\Z_2^3$ occurs exactly two
times \cite{hedayat1999orthogonal}. Moreover, since the distance of $C$ and $(C+\one)$ are both $4$, we obtain that 
\begin{equation}
 \tfrac12(R_{00}^S+R_{11}^S)
 =\frac{I_{2^3}}{2^3},
 \qquad |S|=3.
 \label{eq:binarydiag}
\end{equation}
For the off-diagonal logical block, let $r,s\in\Z_2^S$.  Then
\begin{equation}
\begin{aligned}
 [R_{01}^S]_{r,s}=\frac{1}{2^3}\#\bigl\{(u,v)\in C^2:\;&u_S=r,
 \ (v+\one)_S=s,\\[-1mm]
 &u_B=(v+\one)_B\bigr\}.
\end{aligned}
 \label{eq:binarycrosscount}
\end{equation}
Suppose that $(u,v)\in C^2$ contributes to a matrix element in
Eq.~\eqref{eq:binarycrosscount}.  The partial-trace condition gives
$u_B=(v+\one)_B$.  Set
$w:=u+(v+\one)=\one+u+v$.  Since $C$ is linear,
$u+v\in C$, and hence $w\in C+\one$.  Moreover,
the equality on $B$ implies $w_B=\zero_B$, so
$\supp(w)\subseteq S$.

\begin{figure}[t]
   \centering
   \includegraphics[scale=0.1]{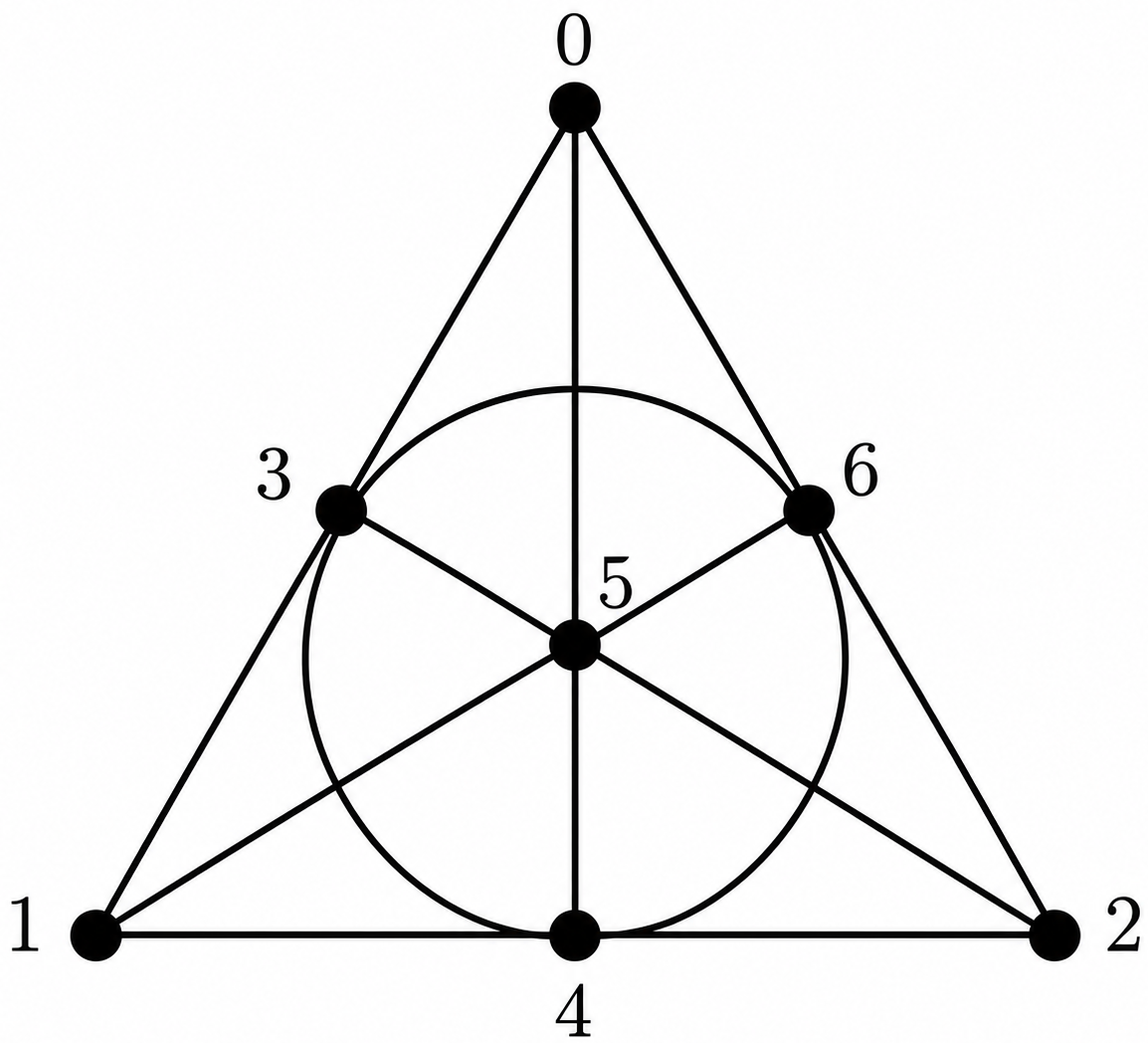}
\caption{The Fano plane in the cyclic labeling used here. Its seven triples are $L_t=\{t,t+1,t+3\}$, $t\in\Z_7$.}
\label{fig:fano}
\end{figure}

Conversely, let $w\in C+\one$ satisfy
$\supp(w)\subseteq S$.  Write $w=\one+c$ with $c\in C$.
For any $u\in C$, set $v=u+c\in C$.  Then
$\one+u+v=w$, and $w_B=\zero_B$ gives
$u_B=(v+\one)_B$.  Thus $(u,v)$ contributes to the matrix element
with retained indices $r=u_S$ and $s=(v+\one)_S$.

The coset $C+\one$ contains seven words of weight three and one word of
weight seven.  Since $|S|=3$, a word in $C+\one$ can have support contained
in $S$ only if it has weight three; in that case its support equals $S$.
The seven weight-three supports are precisely the seven triples of the Fano plane
\begin{equation}
 L_t=\{t,t+1,t+3\}\pmod7,
 \qquad t\in\Z_7,
 \label{eq:Fano}
\end{equation}
shown in Fig.~\ref{fig:fano}; comparison with Table~\ref{tab:binarywords} gives $\supp(w_t)=L_t$ for the seven weight-three coset words.  Since $R_{10}^S=(R_{01}^S)^\dagger$, the support characterization above implies the explicit restriction
\begin{equation}
 R_{01}^S=R_{10}^S=\zero
 \qquad\text{unless } S=L_t \text{ for some }t\in\Z_7.
 \label{eq:binary-cross-support-restriction}
\end{equation}
Thus the binary cross blocks vanish on all $35-7=28$ non-Fano three-subsets.  On a cut $S=L_t$ coming from the Fano plane, Eq.~\eqref{eq:binary-cross-support-restriction} does not assert vanishing; those seven remaining cuts are cancelled by the odd-dimensional sector below.

For odd $q$, let the integer row vector $c=(1,-1,1,0,0,0,0)$ be reduced modulo $q$, and define
\begin{equation}
 |G_a^{(q)}\rangle=q^{-7/2}\sum_{x\in\Z_q^7}
 \omega_q^{Q(x)+a c x^{\mathsf T}}|x\rangle,
 \qquad a\in\Z_2.
 \label{eq:shifted}
\end{equation}
To make the diagonal reductions explicit, define the one-qudit phase unitary
\begin{equation}
 Z_q(\lambda):=\sum_{z\in\Z_q}\omega_q^{\lambda z}|z\rangle\!\langle z|,
 \qquad
 U_a:=\bigotimes_{j=0}^{6} Z_q(a c_j).
 \label{eq:localphase}
\end{equation}
Then Eq.~\eqref{eq:shifted} is equivalently
$|G_a^{(q)}\rangle=U_a|G_q\rangle$.  For the same cut $S|B$, define
\begin{equation}
 T_{ab}^S:=\Tr_B|G_a^{(q)}\rangle\langle G_b^{(q)}|,
 \qquad a,b\in\Z_2.
 \label{eq:graphblocks}
\end{equation}
Write $U_a=U_{a,S}\otimes U_{a,B}$.  Since $|G_q\rangle$ is an
$\AME(7,q)$ state, $|G_a^{(q)}\rangle$ is also an $\AME(7,q)$ state, which implies
\begin{equation}
    T_{aa}^S
 =\frac{I_{q^3}}{q^3},
 \qquad a\in\Z_2.
 \label{eq:diagonalgraphblocks}
\end{equation}
To analyze the cross block $T_{01}^S$, fix a cut $S|B$, $B=S^c$, and write a basis string as $x=(r,y)$, where $r\in\Z_q^S$ and $y\in\Z_q^B$ are row vectors. For $A\subseteq\Z_7$ and a row vector $z\in\Z_q^A$, let
\begin{equation}
 Q_A(z):=\sum_{\substack{i<j\\ i,j\in A}}\Gamma_{ij}z_i z_j\pmod q.
\end{equation}
Then we have
\begin{equation}Q(x)=Q(r,y)=Q_S(r)+Q_B(y)+r\Gamma_{S,B}y^{\mathsf T}.
\end{equation}
For $r,s\in\Z_q^S$, direct partial tracing then gives
\begin{align}
 [T_{01}^S]_{r,s}
 &=q^{-7}\omega_q^{Q_S(r)-Q_S(s)-c_S s^{\mathsf T}}\notag\\[-1mm]
 &\quad\times\sum_{y\in\Z_q^B}
 \omega_q^{\bigl((r-s)\Gamma_{S,B}-c_B\bigr)y^{\mathsf T}},
 \label{eq:cross}
\end{align}
where the term $Q_B(y)$ cancels between ket and bra. By character orthogonality, the sum can be nonzero only if
\begin{equation}
 (r-s)\Gamma_{S,B}=c_B,
 \qquad\text{equivalently}\qquad
 \Gamma_{B,S}(r-s)^{\mathsf T}=c_B^{\mathsf T}.
 \label{eq:crosscondition}
\end{equation}
Consequently, $T_{01}^S=\zero$ whenever $c_B^{\mathsf T}\notin\im(\Gamma_{B,S})$.

For $t\in\Z_7$, define the integer row vector $h_t$ as the cyclic shift of
$h_0=(0,0,1,0,0,-1,1)$ by
\begin{equation}
 (h_t)_j=(h_0)_{j-t\,\bmod 7}.
 \label{eq:cyclicshift-h}
\end{equation}
Since $h_0$ vanishes on $L_0=\{0,1,3\}$, the shifted vector $h_t$
vanishes on $L_t=L_0+t=\{t,t+1,t+3\}$.  Equivalently,
\begin{equation}
 (h_t)_{L_t}=\zero.
 \label{eq:ht-vanishes}
\end{equation}
Put $B_t=L_t^c$.  For $t=0$, use the ordered sets
$L_0=(0,1,3)$ and $B_0=(2,4,5,6)$.  Then
\begin{equation}
 (h_0)_{B_0}=(1,0,-1,1),
 \qquad
 \Gamma_{B_0,L_0}=
 \begin{pmatrix}
 1&0&0\\
 1&1&0\\
 1&1&1\\
 0&1&1
 \end{pmatrix},
 \label{eq:leftnull-t0-data}
\end{equation}
and hence
\begin{equation}
 (h_0)_{B_0}\Gamma_{B_0,L_0}
 =\zero.
 \label{eq:leftnull-t0}
\end{equation}
Because $\Gamma$ is circulant, translating both row and column indices
by $t$ transports this identity to every triple $L_t$ of the Fano plane:
\begin{equation}
 (h_t)_{B_t}\Gamma_{B_t,L_t}=\zero.
 \label{eq:leftnull-all}
\end{equation}
Thus $(h_t)_{B_t}$ is a nonzero left-null vector of
$\Gamma_{B_t,L_t}$.  On the other hand,
\begin{equation}
 \bigl(h_tc^{\mathsf T}\bigr)_{t=0}^6
 =(1,1,-2,2,-1,1,-1),
 \label{eq:witness}
\end{equation}
Decomposing the scalar product according to
$\mathbb Z_7=L_t\cup B_t$ and using
$(h_t)_{L_t}=\zero$ gives
\begin{equation}
\begin{aligned}
 h_tc^{\mathsf T}
 &=(h_t)_{L_t}c_{L_t}^{\mathsf T}
   +(h_t)_{B_t}c_{B_t}^{\mathsf T}\\
 &=(h_t)_{B_t}c_{B_t}^{\mathsf T}.
\end{aligned}
\label{eq:witness-block-decomp}
\end{equation}
Every entry in
Eq.~\eqref{eq:witness} is nonzero modulo an odd $q$.  The seven integer
certificates are displayed in Table~\ref{tab:fanowitness}; the coordinates
of $(h_t)_{B_t}$ are ordered according to the listed ordered complement $B_t$.
\begin{table}[t]
\caption{Fano-cut image-exclusion certificates.  In every
row, $(h_t)_{B_t}\Gamma_{B_t,L_t}=\zero$.}\label{tab:fanowitness}
\centering\small
\setlength{\tabcolsep}{5.0pt}
\begin{ruledtabular}
\begin{tabular}{c c c c c}
$t$ & $L_t$ & $B_t$ & $(h_t)_{B_t}$ & $h_tc^{\mathsf T}$\\
\hline
0 & 013 & 2456 & $(1,0,-1,1)$ & $1$\\
1 & 124 & 0356 & $(1,1,0,-1)$ & $1$\\
2 & 235 & 0146 & $(-1,1,1,0)$ & $-2$\\
3 & 346 & 0125 & $(0,-1,1,1)$ & $2$\\
4 & 045 & 1236 & $(0,-1,1,1)$ & $-1$\\
5 & 156 & 0234 & $(1,0,-1,1)$ & $1$\\
6 & 026 & 1345 & $(1,0,-1,1)$ & $-1$
\end{tabular}
\end{ruledtabular}
\end{table}
For instance, Eqs.~\eqref{eq:leftnull-t0-data} and
\eqref{eq:leftnull-t0} give the $t=0$ certificate explicitly, while
$h_0c^{\mathsf T}=1$.
To prove image exclusion, suppose instead that
$c_{B_t}^{\mathsf T}=\Gamma_{B_t,L_t}z^{\mathsf T}$ for some row vector
$z\in\Z_q^{L_t}$.  Multiplying by the left-null vector gives
\begin{equation}
 h_tc^{\mathsf T}=(h_t)_{B_t}c_{B_t}^{\mathsf T}
 =(h_t)_{B_t}\Gamma_{B_t,L_t}z^{\mathsf T}=0\pmod q,
 \label{eq:imagecontradiction}
\end{equation}
which contradicts Eq.~\eqref{eq:witness}.  Hence
$c_{B_t}^{\mathsf T}\notin\im(\Gamma_{B_t,L_t})$ and
\begin{equation}
 T_{01}^{L_t}=\Tr_{B_t}|G_0^{(q)}\rangle\langle G_1^{(q)}|=\zero
 \qquad(t\in\Z_7).
 \label{eq:graphcross}
\end{equation}

Identify $\C^{2q}\cong\C^2\otimes\C^q$ party by party and set
\begin{equation}
 |\Psi_{2q}\rangle=\frac{|0_L\rangle|G_0^{(q)}\rangle+|1_L\rangle|G_1^{(q)}\rangle}{\sqrt2}.
 \label{eq:coupled}
\end{equation}
The two summands in Eq.~\eqref{eq:coupled} are orthogonal because
$\langle0_L|1_L\rangle=0$, while each basis state has flat amplitudes
and is normalized; hence $|\Psi_{2q}\rangle$ is normalized.  We use the
partywise identification
\begin{equation}
 |b\rangle_2|x\rangle_q\longleftrightarrow|qb+x\rangle_{2q},
 \qquad b\in\Z_2,
 \quad x\in\Z_q,
 \label{eq:localidentification}
\end{equation}
which is a local basis isomorphism and therefore preserves all reduced
spectra.  Using the reduction blocks in Eqs.~\eqref{eq:logicalblocks} and
\eqref{eq:graphblocks}, the three-party marginal is
\begin{equation}
 \rho_S=\frac12\sum_{a,b\in\Z_2}R_{ab}^S\otimes T_{ab}^S.
 \label{eq:block}
\end{equation}
Accorrding to Eqs.~\eqref{eq:binarydiag} and \eqref{eq:diagonalgraphblocks}, the diagonal contribution is
\begin{align}
 \frac12\bigl(R_{00}^S\otimes T_{00}^S+R_{11}^S\otimes T_{11}^S\bigr)
&=\frac12\bigl(R_{00}^S+R_{11}^S\bigr)\otimes\frac{I_{q^3}}{q^3}\notag\\
 &=\frac{I_8}{8}\otimes\frac{I_{q^3}}{q^3}
 =\frac{I_{(2q)^3}}{(2q)^3}.
 \label{eq:diagcombined}
\end{align}
The remaining contribution is the cross term
\begin{equation}
 \rho_S^{\mathrm{cross}}
 :=\frac12\left(
 R_{01}^S\otimes T_{01}^S
 +R_{10}^S\otimes T_{10}^S
 \right).
 \label{eq:cross-contribution}
\end{equation}
The $35$ three-subsets split into the $28$ non-Fano cuts and the seven
Fano cuts $L_t$.  If $S$ is non-Fano, the binary support argument above
gives
\begin{equation}
 R_{01}^S=R_{10}^S=\zero,
 \label{eq:binary-cross-zero}
\end{equation}
and hence $\rho_S^{\mathrm{cross}}=\zero$.  If $S=L_t$ is a
Fano cut, Eq.~\eqref{eq:graphcross} and Hermitian conjugation give
\begin{equation}
 T_{01}^{L_t}=T_{10}^{L_t}=\zero,
 \label{eq:odd-cross-zero}
\end{equation}
so again $\rho_S^{\mathrm{cross}}=\zero$.  Thus the cross
contribution vanishes for every three-subset $S$.  Combining this with
Eq.~\eqref{eq:diagcombined} yields
\begin{equation}
 \rho_S=\frac{I_{(2q)^3}}{(2q)^3}
 \qquad\text{for every }|S|=3,
 \label{eq:all-marginals-twiceodd}
\end{equation}
and therefore
\begin{equation}
 |\Psi_{2q}\rangle\ \text{is an }\AME(7,2q)\text{ state}
 \qquad(q\geq3\text{ odd}).
 \label{eq:twiceodd}
\end{equation}

Use the local basis $|b\rangle_2|x\rangle_q\leftrightarrow|qb+x\rangle_{2q}$.  Eqs.~\eqref{eq:Gsim}, \eqref{eq:shifted}, and \eqref{eq:coupled} give
\begin{align}
 |\Psi_{2q}\rangle
 &=\frac1{4q^{7/2}}
 \sum_{a\in\Z_2}\sum_{m\in\Z_2^3}\sum_{x\in\Z_q^7}
 \omega_q^{Q(x)+a c x^{\mathsf T}}\notag\\[-1mm]
 &\quad\times\bigotimes_{j=0}^6
 \left|q\bigl((mH)_j+a\bigr)+x_j\right\rangle_{2q}.
 \label{eq:explicit}
\end{align}
For $q=3, 5$, Eq.~\eqref{eq:explicit} constructs explicit $\AME(7,6)$ and $\AME(7,10)$, whose existence was open in \cite{AMEtable,RajchelReview2026}.

\textit{\textbf{Proof of the classification.}}\textbf{---}
Note that if $\left|\psi\right\rangle$ is an $\AME(7,d_1)$ state, and $\left|\varphi\right\rangle$ is an $\AME(7,d_2)$ state, then $\left|\psi\right\rangle\otimes \left|\varphi\right\rangle$ is an $\AME(7,d_1d_2)$ state \cite{Rains1999,k_uniform_masking}.
 Write $d=2^\alpha m$ with $m$ odd.  Eq.~\eqref{eq:oddresult} covers $\alpha=0$, and Eq.~\eqref{eq:twiceodd} covers $\alpha=1$.  For $\alpha=2$, we can apply the product property to the known $\AME(7,4)$ state~\cite{Raissi2020,Huang2021} and $\AME(7,m)$.  For $\alpha\geq3$, the resulting $[7,3,5]_{2^\alpha}$ generalized Reed-Solomon code gives a minimal-support $\AME(7,2^\alpha)$~\cite{macwilliams1977theory,Raissi2020}.  Combining it with the odd-dimensional state covers every remaining $d$.  Together with the nonexistence at $d=2$~\cite{Huber2017}, this proves

\begin{table}[t]
\caption{Coverage after writing $d=2^\alpha m$ with $m$ odd.}\label{table:all}
\centering\small
\setlength{\tabcolsep}{5.0pt}
\begin{ruledtabular}
\begin{tabular}{c c}
$\alpha$ & input construction\\
\hline
$0$ & cyclic state $|G_m\rangle$\\
$1$ & coupled state $|\Psi_{2m}\rangle$\\
$2$ & known $\AME(7,4)$ and product property\\
$\geq3$ & Reed--Solomon state and product property
\end{tabular}
\end{ruledtabular}
\end{table}

\begin{theorem}[Complete seven-party classification]
 An $AME(7,d)$ state exists if and only if 
$d\geq3$.
\end{theorem}
See Table~\ref{table:all} for a summary of the constructions. The preceding case split proves existence.  Necessity consists only of the
known nonexistence at $d=2$ \cite{Huber2017}; the local dimension $d=1$ is trivial and is
excluded from the usual AME convention.

For completeness, the input used when $\alpha\geq3$ can be written explicitly.  Put $d_0=2^\alpha\geq8$, choose distinct $\xi_0,\ldots,\xi_6\in\F_{d_0}$, and define
\begin{equation}
 |\Phi_{d_0}\rangle=d_0^{-3/2}\sum_{\beta_0,\beta_1,\beta_2\in\F_{d_0}}
 \bigotimes_{j=0}^{6}|\beta_0+\beta_1\xi_j+\beta_2\xi_j^2\rangle.
 \label{eq:MDSseed}
\end{equation}
The evaluation vectors form a $[7,3,5]_{d_0}$ generalized Reed--Solomon code $C$.  Since the dual distance of $C$ is $4$, $|\Phi_{d_0}\rangle$ is an $\AME(7,d_0)$ state with minimum support  \cite{Feng2017,Goyeneche2014,Raissi2020}.

\textit{\textbf{Conclusion and Discussion.}}\textbf{---}
We have established a complete existence classification of seven-qudit absolutely maximally entangled states, proving that they exist if and only if the local dimension satisfies $d\geq 3$. Together with the previously known results summarized in the AME table~\cite{AMEtable}, this provides a complete existence classification of $n$-qudit AME states for $2\leq n\leq 7$. The first remaining open case is the existence of an $\AME(8,4)$ state. Establishing a complete existence classification for the remaining cases with $n\geq 8$ is a fundamental open problem, whose resolution would provide deeper insights into the structure of multipartite entanglement.

\medskip 
\textit{\textbf{Acknowledgments.}}\textbf{---}
F. S. acknowledges funding from the 2026 Basic Start‑up Fund of Sun Yat‑sen University (Grant No. 67000‑12266020). X. Z. acknowledges funding from the Innovation Program for Quantum
Science and Technology under Grant 2021ZD0302902, in part by
NSFC under Grant 12171452 and Grant 12231014, and in part
by the National Key Research and Development Programs of
China under Grant 2023YFA1010200. Q. Z. acknowledges funding from Innovation Program for Quantum Science and Technology via Project 2024ZD0301900, National Natural Science Foundation of China (NSFC) via Project No. 12347104 and No. 12305030, Guangdong Basic and Applied Basic Research Foundation via Project 2023A1515012185, Hong Kong Research Grant Council (RGC) via No. 27300823, N\_HKU718/23, and R6010-23, Guangdong Provincial Quantum Science Strategic Initiative No. GDZX2303007, HKU Seed Fund for Basic Research for New Staff via Project 2201100596. L. L. acknowledges funding from the National Key Research and Development Program of China (Grant No.2024YFB4504004), the National Natural Science Foundation of China (Grant No. 92465202, 62272492), the Guangdong Provincial Quantum Science Strategic Initiative (Grant No. GDZX2303007, GDZX2503001), the Guangzhou Science and Technology Program (Grant No. 2024A04J4892).

The construction of $\AME(7,2q)$ for odd $q$ was developed with the assistance of GPT-5.6. The authors independently verified the argument, checked all details, and take full responsibility for the correctness of the results and the manuscript.

\bibliographystyle{apsrev4-2}
\bibliography{reference}
\end{document}